\documentclass[%
reprint,
amsmath,amssymb,
aps,
pra,
]{revtex4-1}

\allowdisplaybreaks
\usepackage{amsfonts}
\usepackage{amsthm}
\usepackage{graphicx}
\usepackage{epsfig}
\usepackage{mathtools}
\usepackage[hang]{subfigure}
\usepackage{hyperref}
\usepackage{setspace}

\DeclarePairedDelimiter\ceil{\lceil}{\rceil}

\newcommand{\A}{\mathcal{A}}
\newcommand{\B}{\mathcal{B}}

\newtheorem{theorem}{Theorem}
\newtheorem{lemma}[theorem]{Lemma}
\renewenvironment{proof}{\noindent {\bf Proof: }}{\QED\medskip}
\def\QED{{\hspace*{\fill}{\vrule height 1ex width 1ex }\quad} 
 \vskip 0pt plus20pt}

\newcommand{\be}{\begin{equation}}
\newcommand{\ee}{\end{equation}}
\newcommand{\bea}{\begin{eqnarray}}
\newcommand{\eea}{\end{eqnarray}}
\newcommand{\beann}{\begin{eqnarray*}}
\newcommand{\eeann}{\end{eqnarray*}}

\newcommand{\ket}[1]{\vert{#1}\rangle}
\newcommand{\bra}[1]{\langle{#1}\vert}
\newcommand{\unity}{{1\hskip -3pt \rm{I}}}
\newcommand {\Tr}{\text{Tr}}
\newcommand {\diag}{\text{diag}}

\begin{document}

\title{Capacity of a quantum memory channel correlated by matrix product states}

\author{Jaideep Mulherkar}%
 \email{jaideep\_mulherkar@daiict.ac.in}

\author{V Sunitha}
\email{v\_suni@daiict.ac.in.}
\affiliation{Dhirubhai Ambani Institute of Information and Communication Technology, Gandhinagar, India.}

\begin{abstract}
We study the capacity of a quantum channel where channel acts like controlled phase gate with the control being provided by a one-dimensional quantum spin chain environment.  Due to the correlations in the spin chain, we get a quantum channel with memory. We derive formulas for the quantum capacity of this channel when 	the spin state is a matrix product state. Particularly, we derive exact formulas for the capacity of the quantum memory channel when the environment state is the ground state of the AKLT model and the Majumdar-Ghosh model. We find that the behavior of the capacity for the range of the parameters is analytic.
\end{abstract}

\maketitle
\section{Introduction}
\label{sec:Introduction}
Quantum memory channels are quantum channels in which the noise effects are correlated. The quantum capacity of a quantum memory channel is the maximum information carrying capacity of the channel. This capacity has an operational meaning described in terms of protocols of entanglement transmission or by the amount of distillable entanglement. Computing the capacity of such a channel is a difficult optimization  task since the capacity formula is not given by a single letter formula. Instances of memory channels that occur in physical situations are an unmodulated spin chain \cite{BBMB2008} and a micro maser \cite{GDF2009}. A systematic study of quantum memory channels and a review of results has been done in \cite{KW2005, CGLM2014}. 

In this work we compute the capacities of quantum channels that have a many body environment. The set up is that Alice transmits a sequence of particles to Bob  and each particle interacts via a unitary with its own environmental particle. The environmental particles themselves are in the ground state of  a quantum many body environment like a spin chain. Due to the correlations in the chain, memory effects are introduced. Consequently, the resulting channel between Alice and Bob which is given by tracing out the environment, is a quantum memory channel. This set up was studied in detail in \cite{PV2007,PV2008} wherein conditions were derived under which the quantum capacity equals the regularized coherent information of the channel. Along with this it was also shown that a non-analytic behavior of the capacity was related to a phase-transition in the many-body environment. 

The work in this paper focuses on the case where the quantum many-body environment is a class of spin chains called matrix product states (MPS) or finitely correlated states \cite{FNW1992}. These states are important from the point of view of both quantum computation and condensed matter theory. It is known \cite{VC2004} that they are universal for measurement based quantum computation. Numerical techniques based on higher dimensional MPS (tensor networks) are also being used to study phenomenon such as high temperature superconductivity and  exotic phases of matter \cite{O2014}. MPS states are easy to describe and under some conditions they have nice properties like existence of a spectral gap and exponentially decaying correlations. The MPS formalism  arose out of work done by Majumdar-Ghosh (MG) \cite{MG1969} and Afflek-Kennedy-Lieb-Tasaki (AKLT) \cite{AKLT1988} on valence bond solid models. The quantum many body  environment for the channel in this paper is the ground states of these models.

The main results in this paper are the derivation of exact formulas for the quantum capacities of the quantum memory channels when the environment is in the ground state of the AKLT model and MG model. The work is organized as follows: In Section~\ref{sec:MPS} we give an overview of the matrix product states in terms of the theory of completely positive maps as formulated in \cite{FNW1992}. In  Section~\ref{sec:ChannelConst} we go over the construction of the memory channel. In Section~\ref{sec:AKLT} we derive the formula for the channel correlated by the AKLT and the MG state.

\section{Quantum spin chains and MPS: The AKLT and MG states }
\label{sec:MPS}
A quantum spin chain $\mathcal{A}_{\mathbb{Z}}$  is a suitable limit  of local spin algebras $\mathcal{A}_{[-n,n]}=\otimes_{j=-n}^n \mathcal{A}_i$ where $\mathcal{A}_i$ is the algebra at each site. A state $\omega$ on $\mathcal{A}_{\mathbb{Z}}$ is specified by local density matrices $\rho_{[1,n]}$ such that for any local observable $A_{[1,n]} \in \mathcal{A}_{[1,n]}$
\begin{equation*}
\omega(A_{[1,n]} ) = \Tr(\rho_{[1,n]}A_{[1,n]}).
\end{equation*}
A state $\omega$ is translation invariant if the following holds
\begin{equation*}
\Tr_{n+1}(\rho_{[1,n+1]}\unity_1 \otimes A_{[2,n+1]}) = \Tr(\rho_{[1,n]}A_{[1,n]}).
\end{equation*}
MPS are a class of translation invariant states over $\mathcal{A}_{\mathbb{Z}}$ that are defined by a triple ($\mathcal{B},\rho,\mathbb{E}$) where $\mathcal{B}$ is a matrix algebra, $\rho$ is a state on $\mathcal{B}$, and $\mathbb{E}:\mathcal{A}\otimes \mathcal{B}\mapsto \mathcal{B}$ is a completely positive unital map which has the Krauss representation
\begin{equation}
\label{eq:MPSCP}
\mathbb{E}(A\otimes B) = \sum_j V_j(A\otimes B)V_j^{\dagger},
\end{equation}
where $V_j:\A \otimes\B \rightarrow \B$.
Note that $\mathbb{E}^{(1)}(A):= \mathbb{E}(A\otimes \unity_{\mathcal{B}})$ defines a completely positive map from $\mathcal{A}$ to $\mathcal{B}$. We define $E^{(n)}$ as maps from $\mathcal{A}^{\otimes n}$ to $\mathcal{B}$ recursively as follows
\begin{equation*}
\mathbb{E}^{(n)}:= \mathbb{E}\circ(\unity_{\mathcal{A}}\otimes\mathbb{E}^{(n-1)}), \quad n\geq 1
\end{equation*}
It was shown in \cite{FNW1992} that given the triplet ($\mathcal{B},\rho,\mathbb{E}$) there exists a translation invariant  state  $\omega$ on $\mathcal{A}_{\mathbb{Z}}$ such that the local density matrices satisfy the condition
\begin{equation*}
\Tr(\rho_{[1,n]}A_{[1,n]}) = \Tr(\rho\mathbb{E}^{(n)}(A_{[1,n]})).
\end{equation*}
Translation invariance of $\omega$ yields the property
\begin{equation*}
\Tr(\rho\mathbb{E}(\unity_{\A}\otimes B))= \Tr(\rho B)
\end{equation*}
When the completely positive map $\mathbb{E}$ in equation (\ref{eq:MPSCP}) is pure, that is, it has only one Krauss operator, then
\begin{equation}
\label{eq:MPSCPPure}
\mathbb{E}(A\otimes B) = V(A\otimes B)V^{\dagger}.
\end{equation}
The state generated herein is called a purely generated MPS. For purely generated MPS there exists  an orthonormal set $\{\ket{\phi_i}\}$ and operators $A_i:\mathcal{B}\mapsto \mathcal{B}$ defined as
\begin{equation*}
V\ket{\phi_i \otimes \psi} = A_i\ket{\psi},
\end{equation*}
\begin{equation*}
V^{\dagger}\ket{\psi} = \sum_{i=1}^d \ket{\phi_i}\otimes A_i^{\dagger}\ket{\psi},
\end{equation*}
for $\ket{\psi} \in \mathcal{A}$.
From equation (\ref{eq:MPSCPPure}) we  get
\begin{equation*}
\mathbb{E}(\ket{\phi_i}\bra{\phi_j}\otimes B) = A_i B A_j^{\dagger},\, \text{for} \, B\in \mathcal{B}.
\end{equation*}
Furthermore the operators $A_i$  satisfy 
\begin{equation}
\label{eq:MPSPure1}
\sum_i A_i A_i^{\dagger} = \unity \quad \text{and} \quad \sum_i A_i^{\dagger} \rho A_i = \rho.
\end{equation}
It can be shown that the local density matrices $\rho_{[1,n]}$ are given by
\begin{eqnarray}
\label{eq:LocalDensity}
\rho_{[1,n]} &=& \sum_{\mathbf{s},\mathbf{t} } \ket{\phi_{\mathbf{s}}}\bra{\phi_{\mathbf{t}}}\Tr(A_{\mathbf{s}}^{\dagger}\rho A_{\mathbf{t}})\\
\text{where} \, \ket{\phi_{\mathbf{s}}} &=& \ket{\phi_{s_1}\otimes\phi_{s_2}\otimes \cdots \phi_{s_n}} \, \text{and} \, A_{\mathbf{t}} = A_{t_1}.A_{t_2}...A_{t_n}\nonumber
\end{eqnarray}
It was shown in \cite{FNW1992} that for purely generated MPS if there exists a unique $\rho$ satisfying equation $(\ref{eq:MPSPure1})$ then the resulting states are pure states over $\mathcal{A}_{\mathbb{Z}}$. Moreover they are ground states of translation invariant finite range Hamiltonians. They also have some interesting properties like exponentially decaying correlation functions and a ground state enery gap. The AKLT Hamiltonian is given by
\begin{equation}
\label{eq:HamAKLT}
H_{\text{AKLT}} = \sum_i \mathbf{S}_i\cdot\mathbf{S}_{i+1} + \frac{1}{3} (\mathbf{S}_i\cdot\mathbf{S}_{i+1})^2. 
\end{equation}
The AKLT Hamiltonian is a spin-1 chain and $\mathbf{S}_i = (S_i^x,S_i^y,S_i^z)$ are spin-1 observables corresponding to observing the spin at site i in the x,y and z directions respectively. The ground state of the AKLT Hamiltonan is unique and we look at the parametrized version of the AKLT state which has MPS representation given by the triplet ($\mathbb{E},\mathcal{B},\rho$) with
\begin{equation*}
\mathcal{B} = M_2(\mathbb{C}), \,  \rho = \frac{\unity_{\mathcal{B}}}{2}
\end{equation*}
and the map $\mathbb{E}$ determined by the operators $\{A_i\}_{i=1}^3$ where
\begin{equation}
\label{eq:OpAKLT}
A_1 = -\sin(\theta)\sigma_z, \, A_2 = \cos(\theta)\sigma^{+}, \, A_3 = -A_2 , \, \theta\in [0,\pi)
\end{equation}  
where $\sigma_i ,\ i\in \{x,y,z\}$ are the Pauli matrices and $\sigma^{+} = \sigma_x+i \sigma_y$. The ground state of the Hamiltonian in expression (\ref{eq:HamAKLT}) corresponds to $\theta = \cos^{-1}(\sqrt{\frac{2}{3}})$.
The Majumdar-Ghosh Hamiltonian is a spin-$\frac{1}{2}$ Hamiltonian given by 
\begin{equation}
\label{eq:HamMG}
H_{\text{MG}}   =  \sum_i \mathbf{\sigma}_i\cdot\mathbf{\sigma}_{i+1} + \frac{1}{2}\mathbf{\sigma}_i\cdot\mathbf{\sigma}_{i+2}
\end{equation}
It has two distinct ground states and an equal superposition of these two ground states has MPS representation with operators $\{A_i\}_{i=1}^2$ given by 
\begin{equation}
\label{eq:OpMG}
A_1 = \begin{pmatrix}0 & 1 & 0\\ 0 & 0 & -\sqrt{g}\\ 0 & 0 & 0\end{pmatrix} ,\, A_2 = \begin{pmatrix}0 & 0 & 0\\ \sqrt{1-g} & 0 & 0\\ 0 & 1 & 0\end{pmatrix}, 
\end{equation}
where $g\in [0,1)$ is a parameter and the ground state of the Hamiltonian (\ref{eq:HamMG}) corresponds to $g =\frac{1}{2}$. The invariant state $\rho$ is 
\begin{equation*}
\rho = \begin{pmatrix}\frac{1-g}{2} & 0 & 0\\0 & \frac{1}{2} & 0\\ 0 & 0 & \frac{g}{2}\end{pmatrix}.
\end{equation*}

\section{Construction of the memory channel and capacity formula}
\label{sec:ChannelConst}
A quantum channel is a completely positive trace preserving (CPTP) map $\Lambda:\mathcal{B}(\mathcal{H}_A)\rightarrow \mathcal{B}(\mathcal{H}_B)$ where $\mathcal{H}_A$ and $\mathcal{H}_B$ are Hilbert spaces of the input and output systems respectively. Every completely positive map has a Krauss form
\bea
\label{eq:CPTP1}
\Lambda(\rho) = \sum_k V_k \rho V_k^{\dagger},
\eea
where the linear operators $V_k:\mathcal{H}_A \rightarrow \mathcal{H}_B$ satisfy $\sum_k V_k^{\dagger}V_k =\unity$. By Stinesprings' dilation theorem there exists a Hilbert space referred to as environment $\mathcal{H}_E$ and an isometry $V:\mathcal{H}_A \rightarrow \mathcal{H}_B \otimes \mathcal{H}_E$ such that the map $\Lambda$ is given by
\bea
\label{eq:CPTP2}
\Lambda(\rho) = \text{Tr}_E(V\rho V^{\dagger}),
\eea
where we have taken the partial trace over the environment system. If we take the partial trace over the system $B$ instead of the environment we get the complementary map $\tilde{\Lambda}:\mathcal{H}_A \rightarrow \mathcal{H}_E$ which is also CPTP
\bea
\label{eq:CPTP3}
\tilde{\Lambda}(\rho) = \text{Tr}_B(V\rho V^{\dagger}).
\eea
Due to equations $(\ref{eq:CPTP1}), (\ref{eq:CPTP2}), (\ref{eq:CPTP3})$ we can write the output state of the complementary channel as
\bea
\label{eq:CompChannel}
\tilde{\Lambda}(\rho) = [\Tr(V_j^{\dagger}\rho V_k)]
\eea

Memory less channels are channels where the noise acts independently on each input state. Multiple uses of a memory less channel $\Lambda$ is given by the tensor product $\Lambda^{\otimes n}$. In situations where the tensor product structure of the multi use of the channel does not hold,  we get a channel with memory. In a quantum memory channel besides the input and output, there is a third system $\mathcal{M}$ which  represents the state of the memory. The channel operates on the input state and the state of the memory yielding an output state and a new state of the memory. Thus a quantum memory channel is represented by a CPTP map $\Phi:\mathcal{A}\otimes \mathcal{M}\rightarrow \mathcal{M} \otimes \mathcal{B}$. The action of n successive uses of this channel $\Phi$ is  given by the channel $\Phi^{(n)}:\mathcal{A}^{\otimes n} \otimes \mathcal{M}\rightarrow \mathcal{M} \otimes \mathcal{B}^{\otimes n}$ where 
\beann
\Phi^{(n)} = (\unity_{\mathcal{B}}^{\otimes n-1} \otimes \Phi)\circ (\unity_{\mathcal{B}}^{\otimes n-2} \otimes \Phi &\otimes& \unity_{\mathcal{A}})\circ\cdots\\
&&\cdots\circ(\Phi\otimes\unity_{\mathcal{A}}^{\otimes n-1}).
\eeann
The final output state can be determined by performing a partial trace over the memory system.

\begin{figure}
\includegraphics[height=8cm, width=8cm]{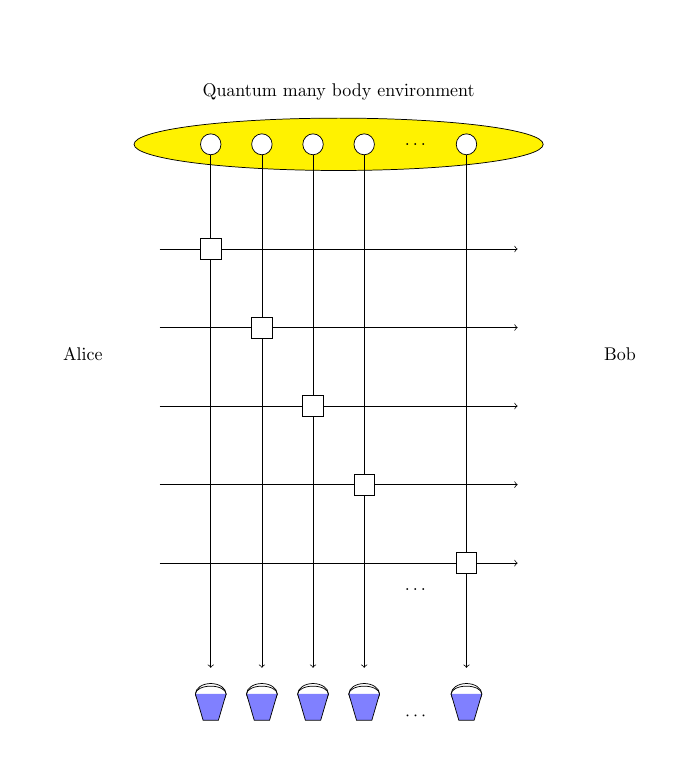}
\caption{Channel construction: Alice transmits a sequence of particles to Bob  and each particle interacts via a unitary (S controlled phase gate) with its own environmental particle. The environmental particles themselves are in the ground state of  a quantum many body environment like a spin chain. Due to the correlations in the chain, memory effects are introduced and the resulting channel which is given by tracing out the environment is a quantum memory channel.}
\label{fig:ChannelConst}
\end{figure}
The memory considered in this paper is a state of a quantum many body system like state of a quantum spin chain. The memory $\mathcal{M}$ consists of a tensor product of spin algebras at each site, that is, $\mathcal{M}=\otimes_{i=1}^n \mathcal{M}_i$ where $\mathcal{M}_i$ is the algebra at site $i$. The memory state is assumed to be initially in an entangled state of the algebra $\mathcal{M}$. The memory channel is created by a family of local unitaries,
\begin{equation*}
U_{q_jM} = \unity_{m_n}\otimes\cdots\otimes\unity_{m_{j+1}}\otimes U_{q_j m_j}\otimes\unity_{m_{j-1}}\otimes \cdots \otimes\unity_{m_1}
\end{equation*}
with each unitary acting on the the channel input and the local spin algebra. The final state of the channel is given by
\begin{eqnarray*}
\Phi^{(n)}(\rho^{(n)}) &=& \Tr_M[U_{q_n m_n}\otimes\cdots \\
&&\cdots \otimes U_{q_1 e_1}(\rho^{(n)}\otimes \omega_M)U^{\dagger}_{q_n m_n}\otimes\cdots \otimes U^{\dagger}_{q_1 e_1}] \nonumber
\end{eqnarray*}
See Figure \ref{fig:ChannelConst}.

We study the special case when the local unitary is a generalized controlled phase gate. For the environment of spin-$\frac{1}{2}$ particles the local unitaries are given by the controlled Z-gate
\begin{equation*}
U_{q_j m_j}= \begin{pmatrix}1 & 0 & 0 & 0\\
               0 & 1 & 0 & 0\\
							 0 & 0 & 1 & 0\\
							 0 & 0 & 0 & -1
\end {pmatrix}.
\end{equation*}
For a d-level environment system the unitary is a generalized phase gate given by
\begin{equation*}
\sum_{k=0}^{d-1} \ket{k}\bra{k}\otimes Z(k),
\end{equation*}
where $Z(k)  =\sum_j \exp(i2\pi kj/d)\ket{j}\bra{j}$. For this channel, if the environment is in state
\begin{equation*}
\sum_{\mathbf{x},\mathbf{y}} \rho_{\mathbf{x},\mathbf{y}}\ket{x_1,x_2,...,x_N}\bra{y_1,y_2,...,y_N}
\end{equation*}
then the action of controlled Z-gate leads to the following dephasing channel
\begin{equation*}
\label{eq:Memchanrep}
\sigma \mapsto \sum_{\mathbf{x}} \rho_{\mathbf{x},\mathbf{x}} Z_1^{x_1}Z_2^{x_2}...Z_N^{x_N}\sigma (Z_1^{x_1}Z_2^{x_2}...Z_N^{x_N})^{\dagger} 
\end{equation*}
Computing the capacities of dephasing channels was also done previously in \cite{DBF2007,LB010,ATL2010}. In \cite{PV2008} it was shown that the quantum capacity of the quantum memory channel with many body environment of a MPS equals the regularized coherent information
\begin{equation*}
Q^c(\Phi) = \lim_{n \rightarrow\infty} \frac{I(\Phi^{(n)})}{n},
\end{equation*}
where 
\begin{equation*}
I(\Phi) = \sup_\rho S(\Phi(\rho)) - S(\tilde{\Phi}(\rho)).
\end{equation*}
Furthermore, in \cite{PV2008}, it was shown that for the quantum memory channel with a many body environment given by a MPS state $\rho_{env}$ the quantum capacity is given by
\begin{equation}
\label{eq:CapacityFormula}
Q^c(\Phi) = \log(d) - \lim_{n\rightarrow\infty} \frac{S(\text{diag}(\rho_{env}))}{n},	
\end{equation}
where $d$ is the dimension of the local spin algebra.

\section{Capacity of channel correlated with the AKLT and MG states}
\label{sec:AKLT}
In the previous section we saw that the capacity of the quantum memory channel can be computed from  the diagonal elements of the local density matrix of the MPS environment. From equation ($\ref{eq:LocalDensity}$) the diagonal elements of local densities $\rho_{(n)}$ of an MPS described by operators $\{A_i\}_{i=1}^l$ and invariant state $\rho$ are given by 
\begin{equation}
\label{eq:diagMPS} 
\text{diag}(\rho_{(n)})= [\Tr(A_{i_n}^{\dagger}A_{i_{n-1}}^{\dagger}...A_{i_1}^{\dagger}\rho A_{i_1}A_{i_2}...A_{i_n})]
\end{equation} 
Now, if $\{A_i\}_{i=1}^l$ are simultaneously diagonalizable, with $A_i = U D_i U^{\dagger}$ for some unitary $U$ then the problem becomes tractable since in this case we get
\begin{equation*} 
\text{diag}(\rho_{(n)})= [\Tr((D_1^{\dagger}D_1)^{k_1}(D_2^{\dagger}D_2)^{k_2}\cdots (D_l^{\dagger}D_l)^{k_l}U \rho U^{\dagger})]
\end{equation*}
where $k_1,k_2,...,k_l$ are the number of occurrences of the operators $A_1,A_2,.., A_l$ in equation (\ref{eq:diagMPS}), respectively. In \cite{PV2008} a solution for this problem was given when the matrices $A_i\otimes A_i^{\dagger}$ are rank-1 matrices by reducing the analysis to the solution of a classical 1D Ising model. If the MPS operators $\{A_i\}_{i=1}^l$are not of the form described in these special cases then the problem is more difficult. Consider the classical probability distribution as
\begin{eqnarray}
\label{eq:classicaldist}  
&& p(X_1=i_1,X_2=i_2....,X_n = i_n) =\\
&& \qquad\qquad\qquad \Tr(A_{i_n}^{\dagger}A_{i_{n-1}}^{\dagger}...A_{i_1}^{\dagger}\rho A_{i_1}A_{i_2}...A_{i_n}) \nonumber
\end{eqnarray}
The translation invariance of the MPS states carries over to this classical probability distribution. The limit $ \lim_{n\rightarrow\infty} \frac{H(X_1,X_2,...,X_n)}{n}$ is the entropy rate of the stochastic process given by the equation (\ref{eq:classicaldist}). Thus, by equation (\ref{eq:CapacityFormula}), computing the diagonal entropy of the MPS state is equivalent to computing the entropy rate of this classical probability distribution. The limit exists when the process is stationary \cite{CT2006} and we have 
\begin{lemma}
\begin{equation*}
 Q^c(\Phi) = \log(d) - \lim_{n\rightarrow\infty} \frac{H(X_1,X_2,...,X_n)}{n}.
\end{equation*} 
\end{lemma} 
A closed form formula for the entropy rate of a stationary stochastic process exists when a stationary process is Markov \cite{CT2006} or a hidden Markov process of a special kind \cite{MMN2012}. In general, obtaining a closed form solution for the entropy rate of a general stationary stochastic process is known to be hard. In what follows we compute the entropy rate of the stationary stochastic process generated by the diagonal elements of the AKLT state and the MG state by explicitly computing the diagonal elements.
Let the local densities of the ground state of the AKLT model and MG model  be denoted by $\rho_{\text{AKLT}}^{(n)}(\theta)$ and $\rho_{\text{MG}}^{(n)}(\theta)$respectively. For AKLT model we have the following lemma
\begin{lemma}
\label{lem:AKLTlemma}
The non-zero diagonal elements of  $\rho_{\text{AKLT}}^{(n)}(\theta)$ are given by $\lambda_p, \, p \in \{0..n\}$ where 
\begin{eqnarray*}
\lambda_p	&=& \frac{\sin^{2p}(\theta) \cos^{2(n-p)}(\theta)}{2}  \\
&&\text{with multiplicity $2\binom{n}{p}$}\text{for}\, p \in \{0..n-1\}\\
\lambda_n &=& \sin^{2n}(\theta) \quad \text{with multiplicity 1}
\end{eqnarray*}
\end{lemma}
\begin{proof}
From equation (\ref{eq:LocalDensity}) and the cyclicity of the trace the diagonal elements of local density matrices $\rho_{\text{AKLT}}^{(n)}(\theta)$ are given by
\begin{equation}
\label{eq:diagelement}
\Tr(\rho O_1O_2...O_n O_n^{\dagger}O_{n-1}^{\dagger}...O_1^{\dagger}) 
\end{equation}
where $\rho = \frac{\unity}{2}$ and the operators $O_i \in \{A_1,A_2,A_3\}$ are as in equation (\ref{eq:OpAKLT}). We denote $\gamma_1 = \unity$ , $\gamma_2 = \begin{pmatrix} 1& 0 \\ 0 &0 \end{pmatrix}$ and $\gamma_3 = \begin{pmatrix} 0 & 0 \\ 0 & 1 \end{pmatrix}$ and note the following relations:
\begin{equation}
\label{eq:AKLTrel1}
A_1 A_1^{\dagger} = \sin^2(\theta)\gamma_1,
\end{equation} 
\begin{equation}
\label{eq:AKLTrel2}
A_2A_2^{\dagger} = \cos^2(\theta)\gamma_2,
\end{equation}
\begin{equation}
\label{eq:AKLTrel3}
A_3A_3^{\dagger} = \cos^2(\theta)\gamma_3, 
\end{equation}
\begin{equation}
\label{eq:AKLTrel4}
A_2^2 = 0, \,  A_3^2 = 0,
\end{equation}
\begin{eqnarray}
\label{eq:AKLTrel5}
&&A_1O_1...O_k O_k^{\dagger}...O_1^{\dagger}A_1^{\dagger} = A_1A_1^{\dagger} O_1...O_k O_k^{\dagger}...O_1^{\dagger} \\ 
&&\qquad\qquad  O_i \in \{A_1,A_2,A_3\} ,\,\ k\in \{1..n\},\nonumber
\end{eqnarray}
\begin{equation}
\label{eq:AKLTrel6}
A_2\gamma_2 A_2^{\dagger} =\cos^2(\theta)\gamma_3, \quad A_3\gamma_3 A_3^{\dagger} =\cos^2(\theta)\gamma_2. 
\end{equation}
Let $p$ be the number of occurrences of operator $A_1$ in equation (\ref{eq:diagelement}), then because of equations (\ref{eq:AKLTrel1}) and (\ref{eq:AKLTrel5}) we have 
\begin{eqnarray}
&&\Tr(\rho O_1 O_2...O_n O_n^{\dagger}O_{n-1}^{\dagger}...O_1^{\dagger}) =\\
&&\qquad  \Tr(\rho \underbrace{A_1A_1^{\dagger}...A_1A_1^{\dagger}}_{\text{p times}}O_{i_1}...O_{i_{n-p}}O_{i_{n-p}}^{\dagger}...O_{i_1}^{\dagger}), \nonumber
\end{eqnarray} 
where the operators $O_{i_k}$ is $A_2 \, \text{or} \,A_3$. From (\ref{eq:AKLTrel1}) we get 
\begin{eqnarray}
\label{eq:AKLTrel7}
&&\Tr(\rho O_1 O_2...O_n O_n^{\dagger}O_{n-1}^{\dagger}...O_1^{\dagger}) =\\
&& \qquad\quad \Tr(\rho \sin^{2p}(\theta)\unity O_{i_1}...O_{i_{n-p}}O_{i_{n-p}}^{\dagger}...O_{i_1}^{\dagger}).\nonumber
\end{eqnarray}
Due to relation (\ref{eq:AKLTrel4}) the  operators $A_2$ and $A_3$ must occur alternately in the sequence  $O_{i_1}...O_{i_{n-p}}$ on the right side of equation (\ref{eq:AKLTrel7}). Now, using equations (\ref{eq:AKLTrel2}),(\ref{eq:AKLTrel3}) and (\ref{eq:AKLTrel6}) we get
\begin{eqnarray*}
&&\Tr(\rho O_1 O_2...O_n O_n^{\dagger}O_{n-1}^{\dagger}...O_1^{\dagger}) =\\
&& \qquad \quad\Tr(\rho \sin^{2p}(\theta)\cos^{2(n-p)}(\theta)\gamma_i)\nonumber
\end{eqnarray*}
where $\gamma_i = \gamma_1$ if $p=n$ otherwise $\gamma_i$ is either $\gamma_2$ or $\gamma_3$. If $p=n$ we get exactly one diagonal element
\begin{equation*}
\lambda_n = \sin^{2n}(\theta). 
\end{equation*}
If $p<n$ then we get diagonal elements of the form 
\begin{equation*}
\lambda_p	= \frac{\sin^{2p}(\theta) \cos^{2(n-p)}(\theta)}{2}.
\end{equation*} 
If $p < n$ then the number of ways that one can place these $A_1$ operators in the sequence $O_1 O_2...O_n$ is $\binom{n}{p}$. Since the remaining  $A_2$ and $A_3$ operators occur alternately, there are only $2$ ways to place them and hence the multiplicity of  $\lambda_p$ in this case is $2\binom{n}{p}$.
\end{proof}
\begin{theorem}
The quantum capacity of the channel correlated with the AKLT ground state is given by
\begin{equation*}
Q^c_{\text{AKLT}} = \log_2 3 - h_2(\theta)
\end{equation*} 	
where $h_2(\theta) = -\sin^2(\theta)\log(\sin^2(\theta))  -\cos^2(\theta)\log(\cos^2(\theta))$ 
\end{theorem}
\begin{proof}
From equation ( \ref{eq:CapacityFormula}) we have 
\begin{equation*}
Q^c_{\text{AKLT}} = \log_2 3 - \lim_{n\rightarrow\infty}\frac{1}{n}\big[-\sum_{k=0}^n \lambda_k \log_2 \lambda_k\big] 
\end{equation*}
where $\lambda_k$ is as in lemma \ref{lem:AKLTlemma}. Setting $\sin^2(\theta) = p$ and $\cos^2(\theta)=q$ we get
\begin{eqnarray*}
&& Q^c_{\text{AKLT}} =\\ 
&& \quad\log_2 3 + \lim_{n\rightarrow\infty}\frac{1}{n} \sum_{k=0}^{n-1}\Big[ 2\binom{n}{k} \frac{p^kq^{n-k}}{2}\log_2(\frac{p^{k}q^{n-k}}{2}) +\\
&& \qquad\qquad\qquad\qquad\qquad\qquad p^n\log_2p^n \Big]
\end{eqnarray*}
The term $\frac{p^n\log_2p^n}{n}$ goes to zero as $n\rightarrow\infty$. Hence, we get
\begin{eqnarray*} 
&&Q^c_{\text{AKLT}} = \log_2 3 + \lim_{n\rightarrow\infty}\frac{1}{n}\Big[\sum_{k=0}^n k \binom{n}{k}p^k q^{n-k}\log_2 p + \\
									&&\sum_{k=0}^n (n-k) \binom{n}{k}p^k q^{n-k}\log_2 q - \sum_{k=0}^{n-1} \binom{n}{k}p^k q^{n-k}\Big] \\
									&=& \log_2 3 + \lim_{n\rightarrow\infty}\frac{1}{n}\Big[\log_2 p\sum_{k=0}^n k \binom{n}{k} p^kq ^{n-k} +\\ 
									 &&  \log_2 q\sum_{k=0}^n (n-k) \binom{n}{k}p^k q^{n-k} - \sum_{k=0}^{n-1} \binom{n}{k}p^k q^{n-k} \Big]\\
									&=&  \log_2 3 + \lim_{n\rightarrow\infty}\frac{1}{n}\Big[\log_2 p\sum_{k=0}^n k \binom{n}{k} p^kq ^{n-k} + \\  && \qquad\qquad \log_2 q\sum_{j=0}^n j \binom{n}{k}q^j p^{n-j} - \sum_{k=0}^{n-1} \binom{n}{k}p^k q^{n-k} \Big]\\
									&=& \log_2 3 + \lim_{n\rightarrow\infty}\frac{1}{n}\Big(np\log_2 p  + 	nq\log_2 q)  - (1-p^n)\Big).\\
\end{eqnarray*}
The last equality follows from the fact that expectation of a binomial random variable with parameters $(n,p)$ is $np$ and the fact that the sum of the binomial probabilities is $1$. Hence
\begin{eqnarray*}
Q^c_{\text{AKLT}}	&=& \log_2 3 + p\log_2 p  +q \log_2 q - \lim_{n\rightarrow\infty} \frac{1-p^n}{n} \\
									&=& \log_2 3 - h_2(\theta)
\end{eqnarray*}
In particular, the capacity of the channel for the ground state of the AKLT Hamiltonian given by equation (\ref{eq:HamAKLT}) is $\log_2 3 - h_2(\cos^{-1}{\sqrt{\frac{2}{3}}}) = \frac{2}{3}$.
\end{proof}

\begin{lemma}
\label{lem:MGlemma}
The non-zero diagonal elements  of  $\rho_{\text{MG}}^{(n)}(\theta)$ are given by:\\ 
\underline{If $n$ is odd}:
\begin{eqnarray*}
\mu_0 &=& \frac{(1-g)^{\ceil*{\frac{n}{2}}} +g^{\ceil*{\frac{n}{2}}}}{2} \quad \text{with multiplicity $2$}\\
\mu_i &=& \frac{(1-g)^{i}g^{\ceil*{\frac{n}{2}}-i}}{2} \quad \text{with multiplicity $2\binom{\ceil*{\frac{n}{2}}}{\ceil*{\frac{n}{2}} -i}$}\\
 &&\qquad\qquad\qquad\qquad\qquad \text{for $1\leq i \leq \ceil*{\frac{n}{2}}-1$}\\
\end{eqnarray*}
\underline{If $n$ is even}:
\begin{eqnarray*}
\gamma_0 &=& \frac{(1-g)^{\frac{n}{2}} +g^{\frac{n}{2}+1}}{2} \quad \text{with multiplicity $1$}\\
\gamma_i^{(1)} &=& \frac{g^i(1-g)^{\frac{n}{2}-i}}{2} \quad \text{with multiplicity $\binom{\frac{n}{2}}{i}$}\\
&&\qquad\qquad\qquad\qquad\qquad\text{for $1\leq i \leq \frac{n}{2}-1$}\\
\gamma_i^{(2)} &=& \frac{g^{\frac{n}{2}-i+1}(1-g)^i}{2} \quad \text{with multiplicity $\binom{\frac{n}{2}+1}{i}$}\\
&& \qquad\qquad\qquad\qquad\qquad \text{for $1\leq i \leq \frac{n}{2}$}\\
\gamma_{\frac{n}{2}} &=& \frac{(1-g)^{\frac{n}{2}+1} +g^{\frac{n}{2}}}{2} \quad \text{with multiplicity $1$}
\end{eqnarray*}
\end{lemma}
\begin{proof}
See appendix.
\end{proof}
\begin{theorem}
The quantum capacity of the channel correlated with the Majumdar-Ghosh ground state is given by
\begin{equation*}
Q^c_{\text{MG}} = 1 + \frac{g}{2}\log_2 (g) + \frac{1-g}{2}\log_2 (1-g) 
\end{equation*} 	
\end{theorem}
\begin{proof}
We consider the cases of $n$ being odd and even separately.\\
\underline{Case 1:} n is odd \\
Setting $k=\ceil*{\frac{n}{2}}$ from lemma \ref{lem:MGlemma} we get
\begin{eqnarray*}
&& Q^c_{\text{MG}} =\\
&&\log_2 2 +\lim_{n\rightarrow\infty}\frac{1}{n}\Big[ 2\frac{(1-g)^k +g^{k}}{2}\log_2 \frac{(1-g)^k +g^{k}}{2}  + \\
&& \qquad\qquad\sum_{i=1}^k 2\binom{k}{i}\frac{g^i(1-g)^{k-i}}{2}\log \frac{g^i(1-g)^{k-i}}{2}\Big]
\end{eqnarray*}
Rearranging the terms we get
\begin{eqnarray*}
&=& 1 + \lim_{n\rightarrow\infty}\frac{1}{n}\Big[\sum_{i=0}^{k} \binom{k}{i}\frac{g^i(1-g)^{k-i}}{2}\\
&&\qquad\qquad\log \frac{g^i(1-g)^{k-i}}{2}\Big]\\
&=& 1 + \lim_{n\rightarrow\infty}\frac{1}{n}\Big[ \log_2g\sum_{i=0}^k i \binom{k}{i} g^i (1-g) ^{k-i} + \\  
&& \qquad \log_2 (1-g)\sum_{i=0}^k (k-i) \binom{k}{i}g^i (1-g)^{k-i} -  \\
&& \qquad\qquad \log_2 2\sum_{k=0}^{k} \binom{k}{i}g^i (1-g)^{k-i} \Big] 
\end{eqnarray*}
The summation in the first two terms inside the limit is expectation of a binomial random variable with parameters $(k,g)$ and $(k,1-g)$ respectively. The summation of last term is just the sum of the binomial probabilities which is equal to one. Thus, we have
\begin{eqnarray*}
&& Q^c_{\text{MG}} = \\
&&  1 + \lim_{n\rightarrow\infty}\frac{1}{n}\Big[kg\log_2 g + (1-g)k\log_2 (1-g) - \log_2 2 \Big]\\
&=& 1 + \lim_{n\rightarrow\infty}\frac{1}{n}\Big[\ceil*{\frac{n}{2}}g\log_2 g + \ceil*{\frac{n}{2}}(1-g)\log_2 (1-g) \Big]\\
&=& 1 + \frac{g}{2}\log_2 g + \frac{1-g}{2}\log_2 (1-g).
\end{eqnarray*}
\underline{Case 2:} n is even \\
Using lemma \ref{lem:MGlemma}, we get
\begin{eqnarray*}
Q^c_{\text{MG}} &=& \log_2 2 +\lim_{n\rightarrow\infty}\frac{1}{n}\Big( \gamma_0\log_2 \gamma_0 +\sum_{i=1}^{\frac{n}{2}-1} \gamma_i^{(0)} \log_2 \gamma_i^{(0)} +\\
&&\qquad\qquad \sum_{i=1}^{\frac{n}{2}} \gamma_i^{(1)} \log_2 \gamma_i^{(1)} +  \gamma_{\frac{n}{2}} \log_2 \gamma_{\frac{n}{2}}\Big)\\
\end{eqnarray*}
Rearranging the terms of $\gamma_0$ and $\gamma_{\frac{n}{2}}$ we can write this as
\begin{eqnarray*}
&& Q^c_{\text{MG}} =\\
&& 1 + \lim_{n\rightarrow\infty}\frac{1}{n}\Big[\sum_{i=0}^{\frac{n}{2}} \binom{\frac{n}{2}}{i}\frac{g^i(1-g)^{\frac{n}{2}-i}}{2}\log_2 \frac{g^i(1-g)^{\frac{n}{2}-i}}{2} + \\
&&\qquad \sum_{i=0}^{\frac{n}{2} +1} \binom{\frac{n}{2} +1}{i}\frac{g^{\frac{n}{2}-i+1}(1-g)^i}{2} \log_2 \frac{g^{\frac{n}{2}-i+1}(1-g)^i}{2}\Big]\\
\end{eqnarray*}
Proceeding along the lines of the odd $n$ case we see that 
\begin{eqnarray*}
Q^c_{\text{MG}} &=& 1 + \lim_{n\rightarrow\infty}\frac{1}{n}\Big[\frac{n}{4} (g \log_2 g + (1-g) \log_2 (1-g) + 1) \\
&& + \quad  \frac{n+1}{4}\big(g \log_2 g + (1-g) \log_2 (1-g) + 1\big)\Big] \\
&=& 1 + \frac{g}{2}\log_2 (g) + \frac{1-g}{2}\log_2 (1-g)
\end{eqnarray*}
\end{proof}
For $g=\frac{1}{2}$ (Majumdar-Ghosh Hamiltonian) given by equation (\ref{eq:HamAKLT}) the capacity is $1 -\frac{1}{4} -\frac{1}{4} = \frac{1}{2}$.

\section{Conclusion}
In this paper we have computed the quantum capacity of a quantum memory channel where the noise correlations come from a spin chain external environment. In particular the environment that we consider are parametrized ground states of the AKLT and  the MG spin chains. These are well known matrix product states and we have used the formalism and mechanics of MPS to compute exact formulas for the capacities of these quantum memory channels. We also observed that the diagonal elements of translation invariant MPS correspond to the probability distribution of a stationary stochastic process. A question that naturally arises is whether the diagonal elements of an MPS correspond to some special stochastic processes that will make the computation of the quantum capacity of the correlated memory channels tractable. Another  question of interest is whether the methods in this work can be used to compute the quantum capacity of a much larger class of quantum memory channels correlated by MPS. The capacity formulas that we derived were analytic in  the range that we have considered. It would be interesting to investigate the effect on the capacities for other ranges of parameter values, especially where there are known phase transitions. 
\newpage
\appendix
\section{Proof of lemma \ref{lem:MGlemma}}
The lemma can be easily verified for the cases when  $n=2$, and $n=3$.  So we prove for the case $n \geq 4$ with $\{A_1,A_2\}$ as in equation (\ref{eq:OpMG}).\\  
\underline{Proof for values:} \\
The proof is based on the observation that, for $n \geq 4$, 
\begin{eqnarray*}
&&\{O_1 \cdots O_nO_n^\dagger  \cdots O_1^\dagger: O_i \in \{A_1, A_2\}, 1 \leq i \leq n \}  = \\
&& \qquad \Big\{ \diag(0,0,0), \diag(g^{\lfloor \frac{n}{2} \rfloor}, 0, 0), \\
&& \qquad\,\,\,\diag(g^i(1-g)^{\lfloor\frac{n}{2}\rfloor - i}, 0, 0): 1 \leq i \leq \lfloor \frac{n}{2} \rfloor - 1,\\
&& \qquad\,\,\,\diag((1-g)^{\lfloor \frac{n}{2}\rfloor}, g^{\lceil \frac{n}{2} \rceil}, 0), \\
&& \qquad\,\,\, \diag(0, g^i(1-g)^{\lceil\frac{n}{2}\rceil - i}, 0): 1 \leq i \leq \lceil \frac{n}{2} \rceil - 1,\\ 
&& \qquad\,\,\, \diag(0, (1-g)^{\lceil \frac{n}{2} \rceil}, g^{\lfloor\frac{n}{2}\rfloor}), \\
&& \qquad\,\,\,\diag(0,0,g^{\lfloor\frac{n}{2}\rfloor - i}(1-g)^i): 1 \leq i \leq \lfloor\frac{n}{2}\rfloor - 1, \\
&& \qquad\,\,\, \diag(0,0, (1-g)^{\lfloor\frac{n}{2}\rfloor}) \Big\}
\end{eqnarray*}
Consequently, for $n \geq 4$, 
\begin{eqnarray*}
&&\{\Tr(\rho O_1 \cdots O_nO_n^\dagger \cdots O_1^\dagger): O_i \in \{A_1, A_2\}, 1 \leq i \leq n \} = \\
&& \qquad\{0, \displaystyle\frac{(1-g)g^{\lfloor\frac{n}{2}\rfloor}}{2}, \\
&& \qquad\,\,\,\displaystyle\frac{g^i(1-g)^{\lfloor\frac{n}{2}\rfloor - i + 1}}{2}: 1 \leq i \leq \lfloor \frac{n}{2} \rfloor - 1,\\
&& \qquad\,\,\,\displaystyle\frac{(1-g)^{\lfloor\frac{n}{2}\rfloor + 1}+g^{\lceil\frac{n}{2}\rceil}}{2}, \\
&& \qquad\,\,\,\displaystyle\frac{g^i(1-g)^{\lceil\frac{n}{2}\rceil - i}}{2}: 1 \leq i \leq \lceil \frac{n}{2} \rceil - 1, \\
&& \qquad\,\,\,\displaystyle\frac{(1-g)^{\lceil\frac{n}{2}\rceil}+g^{\lfloor\frac{n}{2}\rfloor + 1}}{2}, \\
&& \qquad\,\,\,\displaystyle\frac{g^{\lfloor\frac{n}{2}\rfloor - i +1}(1-g)^i}{2}: 1 \leq i \leq \lfloor\frac{n}{2}\rfloor - 1,\displaystyle\frac{g(1-g)^{\lfloor\frac{n}{2}\rfloor}}{2} \} 
\end{eqnarray*}
Now, the non-zero values for $\Tr(\rho O_1 \cdots O_nO_n^\dagger \cdots O_1^\dagger)$ can be inferred as stated in the lemma. The validity of the observation can be established by mathematical induction on $n$.  For $n=4$, the observation can be verified explicitly.  Now, assuming that the observation is valid for an arbitrary integer $k \geq 4$, we prove that the observation holds for $k+1$. 
Let\\
\begin{tabular}{rcl}
$Z_k$ & denote & $\diag_{_{k \times k}} (0,0,0)$, \\
$B_k$ & denote & $\diag_{_{k \times k}} (g^{\lfloor \frac{k}{2} \rfloor}, 0, 0)$, \\
$C_k(i)$ & denote & $\diag_{_{k \times k}} (g^i(1-g)^{\lfloor\frac{k}{2}\rfloor - i}, 0, 0)$, \\
$D_k$ & denote & $\diag_{_{k \times k}} ((1-g)^{\lfloor \frac{k}{2}\rfloor}$, \\
$E_k(i)$ & denote & $\diag_{_{k \times k}} (0, g^i(1-g)^{\lceil\frac{k}{2}\rceil - i}, 0)$, \\
$F_k$ & denote & $\diag_{_{k \times k}} (0, (1-g)^{\lceil \frac{k}{2} \rceil}, g^{\lfloor\frac{k}{2}\rfloor})$, \\
$G_k(i)$ & denote & $\diag_{_{k \times k}} (0,0,g^{\lfloor\frac{k}{2}\rfloor - i}(1-g)^i)$, \\
and $H_k$ & denote & $\diag_{_{k \times k}} (0,0, (1-g)^{\lfloor\frac{k}{2}\rfloor})$.
\end{tabular}

By induction hypothesis 
\begin{eqnarray*}
& &\{O_1 \cdots O_kO_k^\dagger  \cdots O_1^\dagger: O_i \in \{A_1, A_2\}, 1 \leq i \leq k\}  =\\
& & \qquad \{Z_k, B_k, C_k(i): 1 \leq i \leq \lfloor \frac{k}{2} \rfloor - 1,D_k,\\
& & \qquad E_k(i):1 \leq i \leq \lceil \frac{k}{2} \rceil - 1, F_k,\\
& & \qquad G_k(i):1 \leq i \leq \lfloor\frac{k}{2}\rfloor - 1, H_k\}.  
\end{eqnarray*}
So, 
\begin{eqnarray*}
& & \{A_1O_1 \cdots O_kO_k^\dagger  \cdots O_1^\dagger A_1^\dagger ,  A_2O_1 \cdots O_kO_k^\dagger  \cdots O_1^\dagger A_2^\dagger:\\
& & O_i \in \{A_1, A_2\}, 1 \leq i \leq k \}  = \\
& & \{A_1Z_kA_1^\dagger, A_1B_kA_1^\dagger, A_1C_k(i)A_1^\dagger : 1 \leq i \leq \lfloor \frac{k}{2} \rfloor - 1, \\
& & A_1D_kA_1^\dagger, A_1E_k(i)A_1^\dagger :1 \leq i \leq \lceil \frac{k}{2} \rceil - 1, \\
& & A_1F_kA_1^\dagger, A_1G_k(i)A_1^\dagger :1 \leq i \leq \lfloor\frac{k}{2}\rfloor - 1, \\
& & A_1H_kA_1^\dagger, A_2Z_kA_2^\dagger, A_2B_kA_2^\dagger, \\
& & A_2C_k(i)A_2^\dagger : 1 \leq i \leq \lfloor \frac{k}{2} \rfloor - 1, \\
& & A_2D_kA_2^\dagger, A_2E_k(i)A_2^\dagger :1 \leq i \leq \lceil \frac{k}{2} \rceil - 1, \\
& & A_2F_kA_2^\dagger,  A_2G_k(i)A_2^\dagger :1 \leq i \leq \lfloor\frac{k}{2}\rfloor - 1, \\
& & A_2H_kA_2^\dagger\}
\end{eqnarray*}
Now, since 
\begin{eqnarray}
&& A_1Z_kA_1^{\dagger}, A_1B_kA_1^{\dagger}, A_1C_k(i)A_1^{\dagger},A_2Z_kA_2^{\dagger}, A_2G_k(i)A_2^{\dagger}, \nonumber \\
&& A_2H_kA_2^{\dagger} \quad \text{are all}\, Z_{k+1},\, \text{for} \, 1 \leq i \leq \lfloor \frac{k}{2} \rfloor -1,  \\ 
&& A_1D_kA_1^{\dagger} = B_{k+1}, \\
&& A_1 E_k(i) A_1^{\dagger} = C_{k+1}(i)\, \text{for}\, 1 \leq i \leq \lfloor \frac{k+1}{2} \rfloor - 1, \\
&& A_1 F_k A_1^{\dagger} = D_{k+1},\\
&& A_1H_kA_1^{\dagger}, A_2C_k(1)A_2^{\dagger}\, \text{are} \, E_{k+1}(1),\\	
&& A_1G_k (\lfloor \frac{k}{2} \rfloor -i+1)A_1^{\dagger}, A_2 C_k(i)A_2^{\dagger} \, \text{are}\, E_{k+1}(i)\nonumber\\
&& \text{for} \, 2 \leq i \leq \lceil \frac{k+1}{2} \rceil -2, \\
&& A_1G_k(1)A_1^{\dagger}, A_2B_kA_2^{\dagger}\, \text{are}\, E_{k+1}(\lceil \frac{k+1}{2} \rceil - 1),\\
&& A_2 D_kA_2^{\dagger} = F_{k+1},\\
&& A_2 E_k(i)A_2^{\dagger} is G_{k+1}(i)\, \text{for}\, 1 \leq i \leq \lfloor \frac{k+1}{2} \rfloor - 1,\\
&& A_2F_kA_2^{\dagger} = H_{k+1}
\end{eqnarray}
we get 
\begin{eqnarray*}
& & \{A_1O_1 \cdots O_kO_k^\dagger  \cdots O_1^\dagger A_1^\dagger ,  A_2O_1 \cdots O_kO_k^\dagger  \cdots O_1^\dagger A_2^\dagger: \\
& & \qquad O_i \in \{A_1, A_2\}, 1 \leq i \leq k \}  =\\
& & \qquad\qquad \{Z_{k+1}, B_{k+1}, \\ 
& & \qquad\qquad C_{k+1}(i): 1 \leq i \leq \lfloor \frac{k+1}{2} \rfloor - 1, D_{k+1}, \\ 
& & \qquad\qquad E_{k+1}: 1 \leq i \leq \lceil \frac{k+1}{2} \rceil -1,F_{k+1},  \\ 
& & \qquad\qquad G_{k+1}: 1 \leq i \leq \lfloor \frac{k+1}{2} \rfloor - 1, H_{k+1} \}  
\end{eqnarray*}
\underline{Proof of multiplicity:} \\
Let $z_n$, $b_n$, $c_n(i)$, $d_n$, $e_n(i)$, $f_n$, $g_n(i)$, and $h_n$ denote the multiplicity of $Z_n$, $B_n$, $C_n(i)$, $D_n$, $E_n(i)$, $F_n$, $G_n(i)$, and $H_n$, respectively, in $\{ O_1 O_2 \cdots O_n O_n^{\dagger} O_{n-1}^{\dagger} \cdots O_1^{\dagger}: O_i \in\{A_1, A_2\}, 1 \leq i \leq n\}$. From the relations A1-A10 we obtain the following recurrence relations:\\ For $n\geq 5$,
\begin{eqnarray*}
& &z_{n+1}  =  2z_n+b_n+\sum_{i=1}^{\lfloor \frac{n}{2} \rfloor -1} c_n(i) + \sum_{i=1}^{\lfloor \frac{n}{2} \rfloor -1} g_n(i)  + h_n  \\
& & b_{n+1}  =  d_n \\
& & c_{n+1}(i)  =  e_n(i): 1 \leq i \leq \lfloor \frac{n+1}{2} \rfloor - 1, \\
& & d_{n+1}  =  f_n,  \\
& & e_{n+1}(1)   =  h_n + c_n(1), \\
& & e_{n+1}(i)  =  g_n(\lfloor \frac{n}{2} \rfloor - i +1) +c_n(i): 2 \leq i \leq \lceil \frac{n+1}{2} \rceil - 2, \\
& & e_{n+1}(\lceil \frac{n+1}{2} \rceil - 1)  =  g_n(1) + b_n, \\
& & f_{n+1}  =  d_n,  \\
& & g_{n+1}(i)  =  e_n(\lceil \frac{n}{2} \rceil - i): 1 \leq i \leq \lfloor \frac{n+1}{2} \rfloor - 1, \\
& & h_{n+1}  =  f_n
\end{eqnarray*}
with initial conditions
\begin{eqnarray*}
&& z_4 = 6, b_4 = 1, c_4(1)=2, d_4 =  1, e_4 (1)  =  1, \\
&& f_4  =  1, g_4(1)=2,h_4 = 1
\end{eqnarray*}
Solving the above recurrences with the given initial conditions result in 
\begin{eqnarray*}
z_n & = & 2^n - 2^{\lfloor \frac{n}{2} \rfloor + 1} - 2^{\lceil \frac{n}{2} \rceil} + 2, n \geq 4 \\
b_n & = & 1, n \geq 4 \\
c_n(i) & = & \left( \begin{array}{cc} \lfloor \frac{n}{2} \rfloor \\ i \end{array} \right): 1 \leq i \leq \lfloor \frac{n}{2} \rfloor - 1, n \geq 4 \\
d_n & = & 1, n \geq 4 \\
e_n(i) & = & \left( \begin{array}{cc} \lceil \frac{n}{2} \rceil \\ i \end{array} \right): 1 \leq i \leq \lceil \frac{n}{2} \rceil - 1, n \geq 4 \\
f_n & = & 1, n \geq 4 \\
g_n(i) & = & \left( \begin{array}{cc} \lfloor \frac{n}{2} \rfloor \\ i \end{array} \right): 1 \leq i \leq \lfloor \frac{n}{2} \rfloor - 1, n \geq 4 \\
h_n & = & 1, n \geq 4 
\end{eqnarray*}
Now that we have identified the distinct values along with multiplicities for $O_1O_2 \cdots O_nO_n^\dagger O_{n-1}^\dagger \cdots O_1^\dagger$ with $O_i \in \{A_1, A_2\}$ and $1 \leq i \leq n$, we can consider the cases of $n$ being odd or even to conclude the lemma. 
\bibliography{QMFCSChannel}
\end{document}